\documentclass[11pt]{article} 
\usepackage[margin=1in]{geometry} 

\usepackage{graphicx}
\usepackage{url}

\usepackage{xcolor}
\usepackage{physics,nicefrac}
\usepackage{amsmath,amssymb,amsfonts,mathtools}
\usepackage{latexsym,bbm,xspace,graphicx,float}
\usepackage{amsthm}
\usepackage{thm-restate}
\usepackage{thmtools}

\newtheorem{definition}{Definition}
\newtheorem*{theorem*}{Theorem (Informal)}
\newtheorem{theorem}{Theorem}

\newtheorem{corollary}[theorem]{Corollary}

\theoremstyle{remark}

\theoremstyle{remark}

\newcommand{\mx}{\mathbf{x}}
\DeclareMathOperator{\poly}{poly}

\DeclareMathOperator{\supp}{supp}
\newcommand{\eps}{\varepsilon}

\begin{document}

\title{Learning quantum Hamiltonians at any temperature in polynomial time with Chebyshev and bit complexity}
\author{Ales Wodecki and Jakub Marecek} 

\maketitle

\begin{abstract}
We consider the problem of learning local quantum Hamiltonians given copies of their Gibbs state at a known inverse temperature,
following Haah et al.\ [2108.04842] and Bakshi et al.\ [arXiv:2310.02243].
Our main technical contribution is a new flat polynomial approximation of the exponential function based on the Chebyshev expansion,
which enables the formulation of learning quantum Hamiltonians as a polynomial optimization problem. 
This, in turn, can benefit from the use of moment/SOS relaxations, whose polynomial bit complexity requires careful analysis [O'Donnell, ITCS 2017].
Finally, we show that learning a $k$-local Hamiltonian, whose dual interaction graph is of bounded degree, runs in polynomial time under mild assumptions.
\end{abstract}

\section{Introduction}

There is a considerable recent interest in learning models of quantum systems \cite{gebhart2023learning,anshu2023survey}. 
The problem is non-trivial, due to its non-commutativity and non-convexity.
A natural approach to such problems utilizes an approximation of the matrix exponential or matrix logarithm with polynomials 
and transforms the problem into a polynomial optimization problem (POP), either matrix-valued or operator-valued.
Either way, one can obtain the so-called moment/SOS hierarchy \cite{henrion2020moment} of semidefinite programming (SDP) relaxations.
In Theoretical Computer Science, recent papers \cite{anshu2021sample,haah2022optimal,bakshi2023learning} considered the use of Taylor expansion to obtain an approximation of the matrix exponential using polynomials.
Anshu et al.\ \cite{anshu2021sample} considered the sample complexity, without bounding the run-time of the algorithms. 
Haah et al.\ \cite{haah2022optimal} developed the first SDP relaxations with bounds on their dimensions, albeit restricted to the case of high-temperature Gibbs states. 
Bakshi et al.\ \cite{bakshi2023learning} developed the first SDP instances with bounds on their dimensions, without any restriction on the temperature of the Gibbs states,
by showing that there exists a so-called flat approximation of the matrix exponential based on Taylor series. 
In both cases \cite{haah2022optimal,bakshi2023learning}, the authors omit the study of bit complexity, which is important in proving the polynomial run time of an algorithm, especially considering that SDP relaxations obtained from
POP using the moment/SOS method \cite{o2017sos,raghavendra2017bit,gribling2023note} exhibit superpolynomial bit complexity even for any constant-factor approximation, in general.  
Independently in the physics community, \cite{marecek2020quantum,bondar2022recovering,bondar2023globally} have developed the use of Magnus expansion and moment/SOS relaxations in learning models of closed and open quantum systems from estimates of arbitrary states, which have been obtained using quantum state tomography, and
demonstrated the practical performance of this approach. 
Bondar et al.\ \cite[Figure 2]{bondar2023globally} also showed that the use of Chebyshev expansion \cite{Chebyshev1854,Chebyshev1859} is numerically superior to the use of the Taylor expansion. 
This motivated our work that aims to improve the run-time bounds of Bakshi et al.\ \cite{bakshi2023learning}.

In this paper, we show that the use of Chebyshev approximation also leads to instances of a polynomial optimization problem (POP)
for which the SDP relaxations have bounds on their dimensions.
Our main technical contribution is a new flat polynomial approximation to the exponential function, based on the Chebyshev expansion.
As an example of use of this new flat approximation, we also show that SDP relaxations of learning local quantum Hamiltonians given copies of their Gibbs state, can be solved in polynomial time, under mild assumptions. 

\section{A Flat Approximation of the Exponential Using Chebyshev Series}

Polynomial approximations of matrix exponential are key technical tools in the analysis of a number of optimization problems in quantum information theory \cite[e.g.]{marecek2020quantum,anshu2021sample,haah2022optimal,bondar2022recovering,bondar2023globally,bakshi2023learning}.
Bakshi et al.\ \cite{bakshi2023learning} suggested that one should aim to use a so-called flat approximation of the exponential: 

\begin{definition}\label{def_flat_approx}
Let $\epsilon, \eta \in \left(0,1\right)$ and $K>0$ then a polynomial $p$ is called a $\left(\epsilon,\eta,K\right)-$approximation if
\begin{itemize}
    \item $\left|p\left(x\right)-e^{x}\right|\leq\epsilon\text{ for }x\in\left[-K,K\right]$,
    \item $\left|p\left(x\right)\right|\leq\max\left(1,e^{x}\right)e^{\eta\left|x\right|}\label{eq_flat_approx_property}$.
\end{itemize} 
\end{definition}

Unfortunately, the Taylor, Chebyshev, and QSVT-style series fail to provide a flat approximation of the exponential directly \cite{Gil19,tang2023cs,bakshi2023learning}. In \cite{bakshi2023learning}, a flat approximation is constructed based on Taylor series. Motivated by the numerical advantages that a Chebyshev expansion provides over Taylor's \cite[Figure 2]{bondar2023globally}, we present two novel constructions of a polynomial expansion based on Chebyshev's approximation that is flat (Definition \ref{def_flat_approx}). The second definition definition that the approximation must satisfy is a boundedness property described by the following definition.

\begin{definition}\label{def_poly_boundedness}
Let $p\left ( x \right )$ be a polynomial in the variable $x\in\mathbb{R}^{n}$. Then we call a the polynomial $\left(d,C\right)$ bounded if
\begin{itemize}
    \item the degree of $p$ is at most $d$,
    \item for each monomial in $p$ having degree $q \leq d$ the coefficients have magnitue at most $\frac{C}{q!}$.
\end{itemize}
\end{definition}

\subsection{Approximation Properties of Chebyshev Series and Bessel Functions}

For the readers' convenience, we provide a short review of relevant definitions and results related to the Chebyshev expansion and modified Bessel functions of the first kind, which are relevant in the context of the Chebyshev series approximation of the exponential.
\begin{definition}\label{def_cheb_poly} The polynomial functions $T_{n}$, which satisfy the recurrence relations 
\begin{align} 
T_{0}\left(x\right) &=  1, \\ 
T_{1}\left(x\right) &=  x, \\
T_{n+1}\left(x\right) &= 2xT_{n}\left(x\right)-T_{n-1}\left(x\right),\text{ where }n\geq1
\end{align}
are called Chebyschev polynomials of the first kind.

\end{definition}
\begin{definition}
\label{def:Chebyshev}
Let $f:\left[-1,1\right]\rightarrow R$ be an analytical function then 
\begin{equation}
f_{n}\left(x\right)=\sum_{j=0}^{n}a_{j}T\left(x\right),
\end{equation}
where
\begin{equation}\label{eq_coefs_of_chebyshev}
a_{0}=\frac{1}{\sqrt{\pi}}\int_{-1}^{1}f\left(x\right)\frac{dx}{\sqrt{x^{2}-1}},a_{j}=\frac{2}{\sqrt{\pi}}\int_{-1}^{1}f\left(x\right)T_{j}\left(x\right)\frac{dx}{\sqrt{x^{2}-1}}\text{ for }j\geq1    
\end{equation}
is called the Chebyshev series truncation of order $n$.
\end{definition}

\begin{theorem}\label{thm_cheb_est}
Let $T_{n}$ be the series of polynomials from Definition \ref{def_cheb_poly}, then for any $k\geq 0$
\begin{align}\label{eq_chebyshev_upper_bound}
    \left|T_{k}\left(x\right)\right| & \leq \left(x+\sqrt{x^{2}-1}\right)^{k} &\text{ for }\left|x\right|\geq1, \\
    \left|T_{n}\left(x\right)\right| & \leq 1 &\text{ for }\left|x\right|\leq1 
\end{align}
holds.
\end{theorem}
\begin{proof}
The proof is given in the appendix of \cite{tang2023cs}.
\end{proof}

To conclude the overview, we provide a key result, which shows that a truncation of the Chebyshev series provides a uniform approximation of an analytic function on $\left[-1,1\right]$, provided that it possesses a complex analytical continuation.

\begin{theorem} \label{thm_chebyschev_std}
Let $f$ be an analytical function on $\left[-1,1\right]$, which possesses an analytical continuation defined on the Bernstein ellipse $E_{\rho}=\left\{ \frac{1}{2}\left(z+z^{-1}\right):\left|z\right|=\rho\right\}$ on which $\left|f\left(z\right)\right|\leq M$. Then the Chebyshev coefficients satisfy 
\begin{equation}
\left|a_{0}\right|\leq M,\left|a_{k}\right|\leq2M\rho^{-k},\text{ where }k\geq1.
\end{equation}
Consequently, the Chebyshev truncation satisfies
\begin{equation}
    \left\Vert f_{n}-f\right\Vert _{\left[-1,1\right]}\leq\frac{2M\rho^{-n}}{\rho-1}
\end{equation}
and by setting $n$ sufficiently large ($n=\left\lceil \frac{1}{\ln\rho}\ln\frac{2M}{\left(\rho-1\right)\epsilon}\right\rceil $), we get
\begin{equation}
\left\Vert f_{n}-f\right\Vert _{\left[-1,1\right]}\leq\epsilon,
\end{equation}
for any given $\epsilon>0$.
\end{theorem}
\begin{proof}
    See \cite{tang2023cs}.
\end{proof}

\begin{theorem}
Let 
\begin{equation}\label{eq_bessel_def}
I_{v}\left(z\right)\coloneqq i^{-v}J_{v}\left(ix\right)=\sum_{m=0}^{\infty}\frac{1}{m!\Gamma\left(m+v+1\right)}\left(\frac{x}{2}\right)^{2m+v}    
\end{equation}
denote the modified Bessel function, then
\begin{equation}
\label{eq_bessel_upper}
I_{v}\left(x\right)<\frac{\cosh x}{\Gamma\left(v+1\right)}\left(\frac{x}{2}\right)^{v} \quad \text{ for }x>0,v>-\frac{1}{2},
\end{equation}
\begin{equation}\label{eq_bessel_lower}
I_{v}\left(x\right)>\frac{1}{\Gamma\left(v+1\right)}\left(\frac{x}{2}\right)^{v}\text{ for }x>0,v>-\frac{1}{2}    
\end{equation}
holds.

\end{theorem}
\begin{proof}
The definition of \eqref{eq_bessel_def} is standard and \eqref{eq_bessel_upper} and \label{eq_bessel_lower} are shown in \cite{LUKE197241}.
\end{proof}

\begin{theorem}
Let $t \geq 1$, then
\begin{equation}\label{eq_e_x_chebyshev}
e^{x}=\sum_{n=0}^{\infty}a_{n}\left(t\right)T_{n}\left(\frac{x}{t}\right),
\end{equation}
where 
\begin{equation}
a_{0}\left(t\right)=I_{0}\left(t\right),a_{n}\left(t\right)=2I_{n}\left(t\right)\text{ for all }n\geq1.
\end{equation}
\end{theorem}

\begin{proof}
Let $f\left(x\right)\coloneqq e^{tx}$ and applying theorem \ref{thm_chebyschev_std} to $f$. Applying the transformation $x\mapsto\frac{x}{t}$ and considering \eqref{eq_coefs_of_chebyshev}, \eqref{eq_bessel_def} results in \eqref{eq_e_x_chebyshev}. A more detailed exposition can be found in \cite{tang2023cs}.
\end{proof}

\begin{theorem}
Let $N \in\mathbb{N}$ and $f_{N}$ be the truncation of \eqref{eq_e_x_chebyshev} then
\begin{equation} 
\left\Vert f_{N}\left(x\right)-e^{x}\right\Vert _{\left[-t,t\right]}\leq\epsilon,
\end{equation}
whenever $N\geq et+\ln\left(\frac{1}{\epsilon}\right)$.
\end{theorem}
\begin{proof}
See \cite{Gil19} (Lemma 59).
\end{proof}
Note that the previous result shows that the Chybeshev expansion given by \eqref{eq_e_x_chebyshev} exhibits better approximation properties compared to Taylor's expansion, where an approximation of accuracy $\epsilon$ is reached for $N$ satisfying \cite{bakshi2023learning}
\begin{equation}
    N\geq 10t+\ln\left(\frac{1}{\epsilon}\right).
\end{equation}

The flatness property proven in this section hinges on the following recent result \cite{wodecki2024taylorlike}, which gives an algebraic criterion that can be used to show that the Chebyshev series truncation provides upper and lower bounds for the exponential for $x <= -1$.

\begin{theorem}[Reduction theorem]\label{thm_reduction_theorem} Let $N\in \mathbb{N}$ and a let 
\begin{equation}\label{eq_cheby_approx_of_exponential}
f_{n}=\sum_{n=0}^{N}a_{n}T_{n}\left(x\right)   
\end{equation}
be the $N$-th Chebyshev approximation on $L^{2}\left(\left[-1,1\right],\frac{1}{\sqrt{1-x^{2}}}dx\right)$. Then 
\begin{equation}\label{eq_upper_bound_cheb}
e^{x}\leq\sum_{n=0}^{N}a_{n}T_{n}\left(x\right)\end{equation}
holds for even $n$ on $\left(-\infty,-1\right)$ if the polynomial
\begin{equation}\label{eq_key_polynomial_for_thm}
G_{N}\left(x\right)=I_{N+1}\left(1\right)U_{N-1}\left(x\right)+I_{N}\left(1\right)U_{N-2}\left(x\right)-I_{N}\left(1\right)+I_{N}\left(1\right)T_{N}\left(x\right)
\end{equation}
is positive on $\left(-\infty,-1\right)$ and 
\begin{equation}
e^{x}\geq\sum_{n=0}^{N}a_{n}T_{n}\left(x\right)\end{equation}
holds for odd $n$ on $\left(-\infty,-1\right)$ if $G_{N}\left(x\right)$ is negative on $\left(-\infty,-1\right)$.
\end{theorem}
\begin{proof}
See \cite{wodecki2024taylorlike}.
\end{proof}
The criterion can be applied along with Cauchy root estimates and Sturm's theorem to show that the inequalities do indeed hold for $N \leq 1000$. An analogous statement holds in case the interval $\left[-1,1\right]$ is transformed into $\left[-t,t\right]$ for $t > 1$. In the following, it will be assumed that the conditions of the Reduction theorem hold true and that indeed we do have the inequality \eqref{eq_upper_bound_cheb}, which can be checked in polynomial time (quadratic) using the aforementioned method.

\begin{theorem}\label{thm_chebyshev_product_property}
Let $T_{n}$ denote the $n$-th Chebyshev polynomial of the first kind defined in Definition \ref{def_cheb_poly}. Then the following product identity holds
\begin{equation}\label{eq_cheby_product_formula}
T_{m}\left(x\right)T_{n}\left(x\right)=\frac{1}{2}\left(T_{m+n}\left(x\right)+T_{\left|m-n\right|}\left(x\right)\right),\quad\text{for all }m,n\geq0.
\end{equation}
Furthermore, Chebyshev polynomials may be expressed in the monomial basis as
\begin{equation}\label{eq_cheby_in_monomial}
T_{n}\left(x\right)=\frac{n}{2}\sum_{k=0}^{\left\lfloor \frac{n}{2}\right\rfloor }\left(-1\right)^{k}\frac{\left(n-k-1\right)!}{k!\left(n-2k\right)!}\left(2x\right)^{n-2k},\quad\text{for all }n>0.
\end{equation}

\end{theorem}
\begin{proof}
Follows by direct calculation.
\end{proof}

\subsection{Flat Approximations of the Exponential using Chebyshev Series using a Product Construction}\label{sec_chebyshev_product_flat}
Drawing on \cite{bakshi2023learning, Gil19, tang2023cs} we define a product of Chebyshev series the the following way.

\begin{definition}
    Let $k,l\in\mathbb{N}$ then we call
    \begin{equation}\label{eq_flat_approx_def}
Q_{k,l}\left(x\right)\coloneqq f_{2l}\left(\frac{x}{k}\right)f_{4l}\left(\frac{x}{k}\right)\ldots f_{2^{k}l}\left(\frac{x}{k}\right)=\prod_{i=1}^{k}f_{2^{i}l}\left(\frac{x}{k}\right)
    \end{equation}
    the iteratively truncated Chebyshev series approximation of $e^{x}$ with parameters $k$ and $l$.
\end{definition}
Next, the key theorem that shows the conditions under which a polynomial of the form \eqref{eq_flat_approx_def} is a flat approximation
\begin{theorem}\label{thm_flat_approx_theorm}
Let $\eta\in\left(0,1\right)$ and $\epsilon < 1$ be given and let $Q_{k,l}\left(x\right)$ be defined by \eqref{eq_flat_approx_def}. Then $Q_{k,l}\left(x\right)$ is a flat approximation if $k$ and $l$ satisfy
\begin{equation}\label{eq_setting_k}
k=\left\lceil \frac{1}{\eta}\right\rceil 
\end{equation}
and
\begin{equation}\label{eq_setting_l}
l\coloneqq\left\lceil \frac{1}{2}\left(\left(e+1\right)t+\ln\frac{k}{\epsilon}\right)\right\rceil. 
\end{equation}

\end{theorem}
\begin{proof}
Since all of the roots of the polynomials $T_{n}\left(\frac{x}{t}\right)$ are confined to $\left[-t,t\right]$ and 
\begin{equation}
\sum_{n=0}^{N}a_{n}\left(t\right)T_{n}\left(\frac{x}{t}\right)\stackrel{N\rightarrow\infty}{\rightarrow}e^{x}\text{ for all }x\in\mathbb{R}   
\end{equation}
it follows that
\begin{equation}
\sum_{n=0}^{N}a_{n}\left(t\right)T_{n}\left(\frac{x}{t}\right)\leq e^{x}    \quad x \geq t,   
\end{equation}
which results an approximation of the form \eqref{eq_flat_approx_def} always being dominated by $e^{x}$ regardless of $k$ and $l$.

It remains to prove the upper bound for $x \leq -t$. Using Theorem \ref{thm_reduction_theorem} one may make the estimate
\begin{align}
    f_{2^{j}l}\left(\frac{x}{k}\right) &=f_{2^{j}l-1}\left(\frac{x}{k}\right)+a_{2^{j}l}T_{2^{j}l}\left(\frac{x}{kt}\right) \\
    & \leq e^{\frac{x}{k}}+\frac{2}{\left(2^{j}l\right)!}\left(\frac{t}{2}\right)^{2^{j}l}2^{2^{j}l}\left(\frac{x}{tk}\right)^{2^{j}l} \\
    & \leq e^{\frac{x}{k}}+\frac{2}{\left(2^{j}l\right)!}\left(\frac{x}{k}\right)^{2^{j}l}\label{eq_important_est} \\
    & \leq e^{\frac{x}{k}}+\left(\frac{3x}{2^{j}lk}\right)^{2^{j}l},
\end{align}
where the estimate for modified Bessel functions, Chebyshev polynomials and the Stirling approximation were used (see \eqref{eq_bessel_upper} and \eqref{eq_chebyshev_upper_bound}). Defining $j_{0}$ as the smallest possible integer such that $\frac{3x}{2^{j_{0}kl}}<\frac{1}{2}$ implies that 
\begin{equation}
f_{2^{j}l}\left(\frac{x}{k}\right)\leq1 \quad \text{ for all }j\geq j_{0}.\label{eq_first_est}
\end{equation}
Reusing \eqref{eq_important_est} to investigate the terms $j < j_{0}$ we observe that
\begin{equation}
x\geq\frac{2^{j_{1}}lk}{e}\Longrightarrow e^{\frac{x}{k}}\leq\left(\frac{3x}{2^{j_{1}}lk}\right)^{2^{j_{1}}l},
\end{equation}
therefore there exists a $j_{1}$ such that $0<j_{1}<j_{0}$ and
\begin{align}
f_{2^{j}l}\left(\frac{x}{k}\right) &\leq2e^{\frac{x}{k}} &&\text{ for }j\leq j_{1} \label{eq_second_est}    \\
f_{2^{j}l}\left(\frac{x}{k}\right) &\leq\frac{4}{\left(2^{j}l\right)!}\left(\frac{x}{k}\right)^{2^{j}l} && \text{ for }j_{1}<j\leq j_{0}. \label{eq_third_est}
\end{align}
Utilizing \eqref{eq_first_est}, \eqref{eq_second_est} and \eqref{eq_third_est} to overestimate $Q_{k,l}\left(x\right)$ results in
\begin{equation}
Q_{k,l}\left(x\right)\leq2^{j_{0}}\prod_{j=j_{1}}^{j_{0}}\frac{1}{\left(2^{j}l\right)!}\left(\frac{x}{k}\right)^{2^{j}l}=\frac{2^{j_{0}}\left(\frac{x}{k}\right)^{2^{j_{0}}l-2^{j_{1}}l}}{\left(2^{j_{1}}l\right)!\left(2^{j_{1}+1}l\right)!\ldots\left(2^{j_{0}}l\right)!}.\label{eq_alomost_done_flat}
\end{equation}
Since $\frac{2^{j_{0}}}{\left(2^{j_{1}}l\right)!\left(2^{j_{1}+1}l\right)!\ldots\left(2^{j_{0}-1}l\right)}\leq1$ and $\left(2^{j_{0}}l\right)!\geq\left(2^{j_{0}}l-2^{j_{1}}l\right)!$ the right hand side of \eqref{eq_alomost_done_flat} may be reformulated as
\begin{equation}
Q_{k,l}\left(x\right)\leq\frac{\left(\frac{x}{k}\right)^{2^{j_{0}}l-2^{j_{1}}l}}{\left(2^{j_{0}}l-2^{j_{1}}l\right)!}\leq e^{\frac{\left|x\right|}{k}}, \label{eq_flatness_result}
\end{equation}
which shows that $Q_{k,l}\left(x\right)$ is a flat approximation for any $\eta \geq \frac{1}{k}$.

Lastly, a proof of approximation accuracy is provided. Considering an approximation of the form \eqref{eq_flat_approx_def} once again, we may make the estimate

\begin{align}
\left|Q_{k,l}\left(x\right)-e^{x}\right| &=\left|\prod_{j=1}^{j_{0}}f_{2^{j}l}\left(\frac{x}{k}\right)-\underbrace{e^{\frac{x}{k}}e^{\frac{x}{k}}\ldots e^{\frac{x}{k}}}_{k\text{ times}}\right| \\
& \leq\left|f_{2l}\left(\frac{x}{k}\right)\right|\left|\prod_{j=2}^{j_{0}}f_{2^{j}l}\left(\frac{x}{k}\right)-\underbrace{e^{\frac{x}{k}}e^{\frac{x}{k}}\ldots e^{\frac{x}{k}}}_{k-1\text{ times}}\right| \\
& +\left|e^{\frac{k-1}{k}x}\right|\left|f_{2l}\left(\frac{x}{k}\right)-e^{\frac{x}{k}}\right| \\
&\leq\cdots\leq\sum_{j=1}^{k}e^{tj}\left|f_{2^{j}l}\left(\frac{x}{k}\right)-e^{\frac{x}{k}}\right|.
\end{align}

To complete the proof, we set $k$ for a given $\eta$ according to \eqref{eq_flatness_result} i.e. $k=\left\lceil \frac{1}{\eta}\right\rceil $, then by setting 
\begin{equation}
l\coloneqq\left\lceil \frac{1}{2}\left(et+\ln\frac{ke^{t}}{\epsilon}\right)\right\rceil =\left\lceil \frac{1}{2}\left(\left(e+1\right)t+\ln\frac{k}{\epsilon}\right)\right\rceil     
\end{equation}
we ensure that the approximation error is at most $\epsilon$ within $\left[-t,t\right]$.
\end{proof}

Unlike the Taylor based construction \cite{bakshi2023learning}, the coefficients of the flat approximation based on Chebyshev series expansions can not be counted directly, due to the more complex structure of Chebyshev polynomials. To prove the boundedness (in the sense of Definition \ref{def_poly_boundedness}) of a flat approximation given by \eqref{eq_flat_approx_def} an intermediate theorem is required.

\begin{theorem}
Let $N,M \in \mathbb{N}$ such that $N>M$ and let $a_{0},a_{1},\ldots,a_{M}$, $b_{0},b_{1},\ldots,b_{N}$ be sequences or real numbers, then
\begin{align}
\left(\sum_{m=0}^{M}a_{m}T_{m}\right)\left(\sum_{n=0}^{N}b_{n}T_{n}\right)= & \label{eq_product_identity_cheb} \\
\frac{1}{2}\left(a_{0}b_{0}+\sum_{i=1}^{M}a_{i}b_{i}\right)T_{0} & +\frac{1}{2}\left[\sum_{k=1}^{N}\left(\left(\sum_{i=0}^{\min\left\{ k,M\right\} }a_{i}b_{k-i}\right)+\sum_{i=k}^{M}a_{i}b_{i-k}+\sum_{i=k}^{\min\left\{ N,M+k\right\} }a_{i-k}b_{i}\right)T_{k}\right]. \notag 
\end{align}
\end{theorem}
\begin{proof}
Multiplying out the product and using the identity \eqref{eq_cheby_product_formula} yields
\begin{align}
\left(\sum_{m=0}^{M}a_{m}T_{m}\right)\left(\sum_{n=0}^{N}b_{n}T_{n}\right) &= \sum_{m=0}^{M}\sum_{n=0}^{N}a_{m}b_{n}T_{m}T_{n} \\
& =\frac{1}{2}\sum_{m=0}^{M}\sum_{n=0}^{N}a_{m}b_{n}\left(T_{m+n}+T_{\left|m-n\right|}\right),
\end{align}
where the variable was dropped for the sake of readability.
The expression now needs to be arranged by Chebyshev degree. Collecting the terms associated with $T_{m+n}$ and ordering them by degree results in
\begin{equation}\label{eq_proof_of_product_1}
\frac{1}{2}\sum_{k=0}^{N}\left(\sum_{i=0}^{\min\left\{ k,M\right\} }a_{i}b_{k-i}\right)T_{k}=\frac{1}{2}a_{0}b_{0}T_{0}+\frac{1}{2}\sum_{k=1}^{N}\left(\sum_{i=0}^{\min\left\{ k,M\right\} }a_{i}b_{k-i}\right)T_{k}.
\end{equation}
The rearrangement of the coefficients of the terms  $T_{\left|m-n\right|}$ is discussed for the two following cases
\begin{equation}
    \left|m-n\right| = 0 \text{ and } \left|m-n\right|=k\in\mathbb{N}\Rightarrow m-n=k\vee m-n=-k.
\end{equation}
The first condition produces the sum 
\begin{equation}\label{eq_proof_of_product_2}
    \frac{1}{2}\left(\sum_{i=0}^{M}a_{i}b_{i}\right)T_{0},
\end{equation}
while the latter condition lead to the sums
\begin{equation}\label{eq_proof_of_product_3}
\frac{1}{2}\left[\sum_{k=1}^{N}\left(\sum_{i=k}^{M}a_{i}b_{i-k}\right)T_{k}\right] \text{ and } \frac{1}{2}\left[\sum_{k=1}^{N}\left(\sum_{i=k}^{\min\left\{ N,M+k\right\} }a_{i-k}b_{i}\right)T_{k}\right] 
\end{equation}
for $m-n=k$ and $m-n=-k$ respectively. Summing up \eqref{eq_proof_of_product_1}, \eqref{eq_proof_of_product_2} and \eqref{eq_proof_of_product_3} gives \eqref{eq_product_identity_cheb}.
\end{proof}

Building on the previous statement, we are able to show that for the special case in which the approximations are of the form $f_{2^{j}l}\left(\frac{x}{k}\right)$ (see \eqref{eq_flat_approx_def}) the product formula may be simplified. 
\begin{definition}\label{def_super_exp_decay}
We call a series of coefficients $x_{n}$ decaying regularly if there exist constants $C_{1}$ and $C_{2}$ such that
\begin{equation}
\frac{C_{1}}{\Gamma\left(n+1\right)}\left(\frac{t}{2}\right)^{n}\leq x_{n}\leq\frac{C_{2}}{\Gamma\left(n+1\right)}\left(\frac{t}{2}\right)^{n},
\end{equation}
for all $n = 0,1,2, \dots , N$ for some $N\in \mathbb{N}$ and some $t \geq 1$.
\end{definition}
Note that the coefficients of the Chebyshev series expansion of $e^x$ are a super-exponentially decaying sequence of coefficients due to \eqref{eq_bessel_lower} and \eqref{eq_bessel_upper}.

\begin{theorem}\label{thm_product_formula}
Let $a_{n}$ be a sequence of regularly decaying coefficients in the sense of Definition \ref{def_super_exp_decay}. Then there exist constants $h_{i,k}\in\left[0,C\right]$ for each $i,k\in\left\{ 0,1,\ldots,N\right\}$, where $C > 0$ depends on $N$ and the constants $C_1$, $C_2$ governing the exponential decay of the sequence $a_n$ such that
\begin{equation}
\left(\sum_{m=0}^{M}a_{m}T_{m}\right)\left(\sum_{n=0}^{N}a_{n}T_{n}\right)=\sum_{k=0}^{N}\left(\sum_{i=k}^{N}h_{i,k}a_{i}a_{i-l}\right)T_{l}.
\end{equation}
\end{theorem}
\begin{proof}
Consider a product of the form $\left(\sum_{m=0}^{M}a_{m}T_{m}\right)\left(\sum_{n=0}^{N}b_{n}T_{n}\right)$ in which $b_{n} = a_{n}$. Recalling \eqref{eq_product_identity_cheb} one notices that for any given $k$ the terms of $\sum_{i=0}^{\min\left\{ k,M\right\} }a_{i}a_{k-i}$ can be either directly found in $\sum_{i=k}^{M}a_{i}a_{i-k}+\sum_{i=k}^{\min\left\{ N,M+k\right\} }a_{i-k}a_{i}
$ or the exponential decay of the sequence may be used to provide $c_{l,i} \in \left[0,\Tilde{C}\right]$ such that
\begin{equation}
\frac{1}{2}\left[\left(\sum_{i=0}^{\min\left\{ k,M\right\} }a_{i}b_{k-i}\right)+\sum_{i=k}^{M}a_{i}b_{i-k}+\sum_{i=k}^{\min\left\{ N,M+k\right\} }a_{i-k}b_{i}\right]=\sum_{i=k}^{M}c_{k,l}a_{i}b_{i-k}+\sum_{i=k}^{\min\left\{ N,M+k\right\} }c_{k,l}a_{i-k}b_{i},
\end{equation}
where $k\in\mathbb{N}_{0}$, which notably includes the $0$-th terms which was left out of the sum before. Utilizing the symmetry of the two expressions one arrives at
\begin{equation}
\left(\sum_{m=0}^{M}a_{m}T_{m}\right)\left(\sum_{n=0}^{N}a_{n}T_{n}\right)=\sum_{k=0}^{N}\left(\sum_{i=k}^{N}h_{i,k}a_{i}a_{i-l}\right)T_{k},
\end{equation}
where $h_{i,k}\in\left[0,C\right]$ for each $i, k$, where $C = 2 \Tilde{C}$.
\end{proof}

Next, the product formula of Theorem \ref{thm_product_formula} may be used to show that the product of two expansions with regularly decaying sequences is again an expansion with exponentially decaying coefficients.

\begin{theorem}\label{thm_stable_w_r_t_reg}
Let $\left(a_{n}\right)$ be a sequence of regularly decaying coefficients. Then the coefficients of the product
\begin{equation}
\left(\sum_{i=k}^{N}h_{i,k}a_{i}a_{i-l}\right)_{k=0,1,\ldots,N} 
\end{equation}
are also regularly decaying.
\end{theorem}

\begin{proof}
Since the coefficient sequence $\left(a_{n}\right)$ is regularly decaying a transformation of the indices may be used to show that each term in bounded from above by
\begin{align}
h_{i+l,k}a_{i+l}a_{i}\leq\frac{C}{i!\left(i+l\right)!}\left(\frac{t}{2}\right)^{2i+l},    
\end{align}
where $i\in\left\{ 0,1,2,\ldots,N-k\right\}$, where $k$ is the index of the coefficient of the resulting series. Recalling \eqref{eq_bessel_def} and making use of \eqref{eq_bessel_upper} it is clear that 
\begin{equation}
\sum_{i=k}^{N}h_{i,k}a_{i}a_{i-l}\leq CI_{k}\left(t\right)\leq\frac{C\cosh t}{\Gamma\left(k+1\right)}\left(\frac{t}{2}\right)^{k}\leq\frac{2C}{\Gamma\left(k+1\right)}\left(\frac{t}{2}\right)^{k}.
\end{equation}
The lower estimate follows in similar fashion. 
\end{proof}

\begin{corollary}\label{cor_regular_seq}
Let $\left(a_{n}\right)_{n=0,1,\ldots,N_{k}}$ be a regularly decaying sequence (see Definition \ref{def_super_exp_decay}) and let
\begin{equation}
\left(\sum_{m=0}^{N_{1}}a_{m}T_{m}\right)\left(\sum_{n=0}^{N_{2}}a_{n}T_{n}\right)\ldots\left(\sum_{n=0}^{N_{k}}a_{n}T_{n}\right)=\sum_{n=0}^{N_{k}}b_{n}T_{n},
\end{equation}
where $N_{1}\leq N_{2}\leq\dots\leq N_{k}$ for some $k \in \mathbb{N}$. Then $\left(b_{n}\right)_{n=0,1,\ldots,N_{k}}$ is a regulary decaying sequence.
\end{corollary}
\begin{proof}
Follows directly by applying Theorems \ref{thm_product_formula} and \ref{thm_stable_w_r_t_reg} $k-1$ times.
\end{proof}

Finally, the previously derived results are put to use to show the boundedness of a flat approximation given by \eqref{eq_flat_approx_def}.

\begin{theorem}
Let $k, l \in \mathbb{N}$ and $Q_{k,l}$ be given by \eqref{eq_flat_approx_def} then 
is $\left(2^{k+1}l,C_{k,l}\right)$-bounded, where
\begin{equation}
C_{k,l} = ce^{\frac{t^{2}}{4}}\left(\frac{t}{k}\right)^{2^{k}l},    
\end{equation}
where $c$ depends on the regularity coefficients of the sequence $\left(a_{n}\right)_{n=0,1,\ldots,2^{k}l}$ of coefficients of the Chebyshev series expansion. Furthermore, if $t \leq k$ the bound can be made independent of both $l$ and $k$ and becomes
\begin{equation}\label{eq_bounding_constant_simple}
C = ce^{\frac{t^{2}}{4}}.    
\end{equation}
Lastly, if $t < 1$, then again the bounding constant becomes \eqref{eq_bounding_constant_simple} for any $k,l \in \mathbb{N}$.

\end{theorem}
\begin{proof}
Using Corollary \ref{cor_regular_seq} and
\eqref{eq_cheby_in_monomial} to write the resulting polynomial in the monomial basis results in
\begin{align}
\sum_{n=0}^{2^{k}l}b_{n}T_{n}\left(\frac{x}{k}\right) &=\sum_{n=0}^{2^{k}l}b_{n}\frac{n}{2}\sum_{k=0}^{\left\lfloor \frac{n}{2}\right\rfloor }\left(-1\right)^{k}\frac{\left(n-k-1\right)!}{k!\left(n-2k\right)!}\left(2\frac{x}{k}\right)^{n-2k} \\
& =\sum_{i=0}^{2^{k}l}\left(\frac{x}{k}\right)^{i}\sum_{k=0}^{2^{k-1}l-\left\lfloor \frac{i}{2}\right\rfloor }b_{i+2k}\frac{i+2k}{2}\left(-1\right)^{k}2^{i}\frac{\left(i+k-1\right)!}{k!i!} \\
& =\sum_{i=0}^{2^{k}l}x^{i}\frac{1}{i!}\left(\frac{2}{k}\right)^{i}\sum_{k=0}^{2^{k-1}l-\left\lfloor \frac{i}{2}\right\rfloor }b_{i+2k}\frac{i+2k}{2}\left(-1\right)^{k}\frac{\left(i+k-1\right)!}{k!},
\end{align}
where a substitution was applied to resolve expose the coefficients in the monomial basis. Using the regular decay of sequence $\left(b_{n}\right)_{n=0,1,\ldots,2^kl}$ allows for the estimate
\begin{align}
\left|\sum_{k=0}^{2^{k-1}l-\left\lfloor \frac{i}{2}\right\rfloor }b_{i+2k}\frac{i+2k}{2}\left(-1\right)^{k}\frac{\left(i+k-1\right)!}{k!}\right| & \leq\sum_{k=0}^{2^{k-1}l-\left\lfloor \frac{i}{2}\right\rfloor }\frac{1}{\left(i+2k\right)!}\left(\frac{t}{2}\right)^{i+2k}\left|\frac{i+2k}{2}\right|\frac{\left(i+k-1\right)!}{k!} \\
& =c_{1}\sum_{k=0}^{2^{k-1}l-\left\lfloor \frac{i}{2}\right\rfloor }\frac{t^{i+2k}}{2^{i+2k}k!}\frac{\left(i+2k\right)\left(i+k-1\right)!}{2\left(i+2k\right)!}\\
& \leq c_{2}\sum_{k=0}^{\infty}\frac{t^{i+2k}}{2^{i+2k}k!}, \label{eq_final_res_of_bd_in_proof}
\end{align}
where $c_{1}$ is the constant due to regular decay and $c_{2}$ is a consequence of applying
\begin{equation}
\frac{\left(i+2k\right)\left(i+k-1\right)!}{2\left(i+2k\right)!} \leq 1 \text{ for all $i$ and sufficientely large $k$}.
\end{equation}
 Finally, the sum in \eqref{eq_final_res_of_bd_in_proof} may be resolved as
\begin{equation}
\sum_{k=0}^{\infty}\frac{t^{i+2k}}{2^{i+2k}k!}=2^{-i}e^{\frac{t^{2}}{4}}t^{i},
\end{equation}
which shows that the $i$-th monomial coefficient is bounded by
\begin{equation}
c\frac{1}{i!}e^{\frac{t^{2}}{4}}\left(\frac{t}{k}\right)^{i}.
\end{equation}
The statement of the theorem directly follows.

\end{proof}

\subsection{General Setting and Essential Notation}
Assume a $2^{n}$-level quantum system is described by a Hamiltonian of the form
\begin{equation}\label{eq_exact_basic_ham}
H=\sum_{a=1}^{m}\lambda_{a}E_{a},    
\end{equation}
where each $E_{a}$ is a tensor product of Pauli matrices and $\lambda_{a}\in\left[-1,1\right]$. Additionally, assume that all Hamiltonians of the form \eqref{eq_exact_basic_ham} satisfy $\left\Vert E_{a}\right\Vert \leq1,\left|\lambda_{i}\right|\leq1$. Supposing that the system in question is in thermal equilibrium at the inverse temperature $\beta$, we may assume that the density matrix describing the system reads
\begin{equation}\label{eq_gibbs_state}
\rho=\frac{e^{-\beta H}}{\text{Tr}\left(e^{-\beta H}\right)}.
\end{equation}
The polynomial optimization formulation given in the following section presents a way to estimate coefficients $\lambda_{i}$ by coefficients $\widehat{\lambda}_{i}$ with high probability assuming the ability to prepare copies of the Gibbs state \eqref{eq_gibbs_state}.
A key property that allows the polynomial-time identifiability of a system described by \eqref{eq_exact_basic_ham} is the concept of limited interactability of each of the constituents of the Hamiltonian, which is formalized by the following definitions.
\begin{definition}
A Hamiltonian of the form \eqref{eq_exact_basic_ham} is $k$-local if at most $k$ of the terms are non-trivial Pauli matrices. 
\end{definition}
\begin{definition}
A dual interaction graph $G$ associated with a $k$-local Hamiltonian $H$ is an undirected graph with vertices labeled by $\left\{ 1,\ldots,m\right\} $ and edges between vertices $a,b\in\left\{ 1,\ldots,m\right\} $ if $E_{a}$ and $E_{b}$ have at least one qubit on which they both act non-trivially, that is,
\begin{equation}
\supp\left(E_{a}\right)\cap\supp\left(E_{b}\right)\neq\emptyset.
\end{equation}
\end{definition}

\begin{definition}
Let $H$ be a $k-$local Hamiltonian, and $G$ be the dual interaction graph associated with $H$. We say that a tensor product of Pauli matrices $P$ is kl-$G$-local if there exists a set $S\subset\left\{ 1,\ldots,m\right\} $ of size $l$ such that 
\begin{equation}
\supp P\subset\cup_{a\in S}\supp E_{a}.
\end{equation}
Denote the set of all $kl$ local Pauli matrices as $P_{kl}.$
\end{definition}
In order to formulate the measurement condition in polynomial form, one needs to make use of the following correspondence between commutators and polynomials.
\begin{definition}
Let $p\left(x,y\right)=\sum_{i+j\leq d}a_{ij}x^{i}y^{j}$ be a bi-variate polynomial, the associated matrix commutator polynomial reads
\begin{equation}
p\left(X,Y\left|A\right.\right)=\sum_{i+j\leq d}a_{ij}\left[X,\left[Y,A\right]_{j}\right]_{i},    
\end{equation}
where $X,Y,A$ are matrices and $\left[\cdot,\cdot\right]_{i}$ denotes the $i$-the nested commutator. The single variable variant for a polynomial $p\left(x\right)=\sum_{i\leq d}a_{i}x^{i}$ reads
\begin{equation}
p\left(X\left|A\right.\right)=\sum_{i\leq d}a_{i}\left[X,A\right]_{i}.
\end{equation}
\end{definition}

\subsection{Polynomial Optimization Problem Formulation}
With the use of the preceding definitions and the approximation of the trace function detailed in \cite{bakshi2023learning}, which is denoted by $\widetilde{\text{Tr}}$, it is possible to formulate a polynomial optimization problem. Setting $\epsilon > 0$ as the desired error of approximation, we may formulate it as follows:

\begin{align}
    a\in\left\{ 1,\ldots,m\right\}  && \widehat{\lambda}_{a}\in\left[-1,1\right],\nonumber\\
    \forall A_{1},A_{2}\in\mathcal{A} && \left|\widetilde{\text{Tr}}\left(A_{1}A_{2}\left(\widehat{H}\rho-\rho \widehat{H}\right)\right)\right|^{2}\leq\epsilon_{0}^{2},\nonumber \\ \forall B_{1},B_{2}\in\mathcal{B} && \left|\widetilde{\text{Tr}}\left(B_{2}Q_{n,l}\left(-\beta H\left|B_{1}\right.\right)\rho\right)-\widetilde{\text{Tr}}\left(B_{1}B_{2}\rho\right)\right|^{2}\leq\epsilon^{2},\label{eq_poly_opt_problem}
\end{align}
where $\epsilon_{0}=\frac{\epsilon^{10^{C\left(k,G\right)\beta}}}{m^{3}}$, $\mathcal{A}=P_{4C\left(k,G\right)l}$, $\mathcal{B}=P_{4k}$ and $\widehat{H}=\sum_{a=1}^{m}\widehat{\lambda}_{a}E_{a}$. According to \cite{bakshi2023learning} $Q_{n,l}$ must be an at least 

\begin{equation}
\left(10^{-4}2^{C\left(k,G\right)\beta}\ln\left(\frac{1}{\epsilon}\right),\frac{5}{C\left(k,G\right)\beta},10^{-4}\epsilon\right)\text{-flat exponential approximation.} 
\end{equation}
In order to ensure that estimated parameters $\widehat{\lambda}_{a}$ will be at most $\epsilon$ far from $\lambda_{a}$. 
Denote 
\begin{equation}
\delta\coloneqq10^{-4}2^{C\left(k,G\right)\beta}\ln\left(\frac{1}{\epsilon}\right),\eta\coloneqq\frac{5}{C\left(k,G\right)\beta},t\coloneqq10^{-4}\epsilon 
\end{equation}
and $k,l\in \mathbb{N}$ according to \eqref{eq_setting_k} and \eqref{eq_setting_l} then by Theorem \ref{thm_flat_approx_theorm} $Q_{k,l}$ is a flat approximation with the desired properties. Furthermore, since $10^{-4}\epsilon=t<1$ \eqref{eq_bounding_constant_simple} shows that the resulting polynomial is $\left(2^{k+1}l,ce\right)$ bounded, where $c$ depends only on $k$. Taking the previous settings into consideration, one may show that
\begin{equation}\label{eq_boundedness_prop}
\deg\left(Q_{k,l}\right)\leq2^{C\left(k,G\right)\beta+1}\left(e+\ln\left(\frac{1}{\epsilon}\right)\right).    
\end{equation}

Now, we are ready to state one of the main results of our paper. As in Bakshi et al.\ \cite{bakshi2023learning}, who utilized their flat approximation
of the matrix exponential to analyze moment/SOS relaxations closely related to those of Haah et al.\ \cite{haah2022optimal}, we utilize our flat approximation of the matrix
exponential to analyze closely related moment/SOS relaxations:

\begin{theorem}
\label{thm:runtime}
Suppose that $H$ is a $k$-local $n$-qubit Hamiltonian with dual interacting graph $G$. Let $\epsilon > 0$, $\delta >0$ be arbitrary, and let the inverse temperature $\beta$ be known. Then, there exists an algorithm that can estimate Hamiltonian coefficients $\widehat{\lambda}_{a}$ with probability greater than $1-\delta$ with accuracy $\left(\widehat{\lambda}_{a}-\lambda_{a}\right)^{2}\leq\epsilon^{2}$ for all $a\in\left\{ 1,\ldots,m\right\} $. 
Assuming that the bit-complexity of the estimates $\widehat{\lambda}_{a}$ is $\poly(n,\log(1/\eps))$,
the time complexity of the algorithm is
\begin{equation}\label{eq_time_complexity}
\left[m+m^{O\left(1\right)}\left(\frac{1}{\epsilon}\right)^{\exp\left(C\left(k,G\right)\beta\right)}\right]^{E\left(K,\epsilon,G,k\right)},
\end{equation}
where $E\left(K,\epsilon,G,k\right)$ is independent of the size of the problem and it uses 
\begin{equation}\label{eq_number_of_gibbs_cop}
O\left(\left(\frac{m^{6}}{\epsilon^{e^{C\left(k,G\right)\beta}}}\right)\ln\left(\frac{m}{\delta}\right)\right)    
\end{equation}
copies of the Gibbs state.
\end{theorem}

\begin{proof}
The conclustion \eqref{eq_number_of_gibbs_cop} follows from \cite[Theorem 6.1]{Anshu_2021} and setting of problem \eqref{eq_poly_opt_problem}. To prove \eqref{eq_time_complexity}, we note that the number of constraints (by the construction of $\mathcal{A}$ and $\mathcal{B}$)
\begin{equation}\label{eq_est_for_size_of_A_B}
\left|\mathcal{A}\right|+\left|\mathcal{B}\right|\leq m^{O\left(1\right)}\left(\frac{1}{\epsilon}\right)^{\exp\left(O\left(C\left(k,G\right)\beta\right)\right)}.
\end{equation}
Furthermore, \cite{Gro81} showed that if the bit-complexity grows polynomially in the inputs, the ellipsoid algorithm solves \eqref{eq_poly_opt_problem} in time 
\begin{equation}
\left(m+r\right)^{O\left(d\right)},
\end{equation}
where $m$ is the number of variables, $r$ is the number of constraints, and $d$ is the degree of pseudo-distribution that is is output of the algorithm. Using \eqref{eq_boundedness_prop} and \eqref{eq_est_for_size_of_A_B} one arrives at a polynomial complexity with respect to the inputs: 
\begin{equation}
\left[m+m^{O\left(1\right)}\left(\frac{1}{\epsilon}\right)^{\exp\left(C\left(k,G\right)\beta\right)}\right]^{E\left(K,\epsilon,G,k\right)},
\end{equation}
where 
\begin{equation}
E\left(K,\epsilon,G,k\right)=\deg\left(Q_{k,l}\right)\leq2^{C\left(k,G\right)\beta+1}\left(e+\ln\left(\frac{1}{\epsilon}\right)\right)   
\end{equation}
is independent of the input size of the problem.
\end{proof}

\subsection{The Bit Complexity of the Polynomial Optimization Problem}

In the previous section, our Theorem \ref{thm:runtime} assumes that the bit complexity of the estimate is polynomial.
This section discusses this assumption in more detail.

In a string of somewhat negative results on the moment/SOS method for polynomial optimization problems in dimension $n$ \cite{o2017sos,raghavendra2017bit,gribling2023note}, 
it was shown that the existence of SDP relaxations whose dimensions are polynomial in $n$ does not need to guarantee that the bit complexity of the SDP relaxation is polynomial in $n$, or indeed that the runtime of the algorithm for solving the SDP is polynomial in $n$ up to some accuracy.  
However, in the literature on polynomial optimization, it has often been claimed that at a fixed level $d$ in the Moment/Sum-of-Squares (SOS) hierarchy, the
solution of a semidefinite program of size polynomial in the number of variables n (or alternatively, 
the degree $d$ SOS proof) can be obtained in time $O(n^d)$.
However, even when the instance satisfies the so-called ``Archimedean'' condition (slightly stronger than compactness, ``explicitly bounded''),   
there need not be an exponential upper bound on the number of bits needed to write down the solution.
O’Donnell \cite{o2017sos} presents an example of a degree-2 SOS program with bounded coefficients on the input,
but with an exponential number of bits in the solution,
and (in Section 2.1) with an exponential number of bits in any constant-factor approximation of the solution.
Raghavendra et al.\ \cite{raghavendra2017bit} have shown that there are polynomial systems with Boolean constraints and related non-negative polynomials, which have degree-two SOS proofs, but no SOS proof of degree $\Omega(\sqrt{n})$ with small coefficients. 
Pataki et al.\ \cite{pataki2021exponential} explain how such solutions arise in general. 
Thus, any known algorithm, including the ellipsoid algorithm, would take exponential time. 
Such examples crucially rely on the fact that there are equalities in the polynomial optimization problem.

In the case of learning quantum Hamiltonians, where the polynomial optimization problem is not given explicitly,
but is obtained by some combination of Magnus \cite{marecek2020quantum}, Taylor \cite{bakshi2023learning}, and Chebyshev \cite{bondar2023globally} expansions, 
the bit complexity becomes a non-trivial question. 

An approach to this non-trivial question has been recently suggested by Gribling et al.\ \cite{gribling2023note}:

\begin{theorem}[Gribling et al.\ \cite{gribling2023note}]\label{THM:FULLDIM}
Let $S(\mathbf{g}) \subseteq \mathbb{R}^n$ be a semi-algebraic set defined by inequalities. Assume that the following two conditions are satisfied:
\begin{enumerate}
    \item $S(\mathbf{g})$ is explicitly bounded, i.e., there exists a constant $1 \leq R \leq 2^{\mathrm{poly}(n)}$ such that the instance remains feasible after we can add the constraint:
    \begin{align}
    \label{eq:archimed}
    g_1(\mx): \sum_{i=1}^n \mx_i^2 \le R^2.
    \end{align}
    \item $S(\mathbf{g})$ contains a ball of radius $r \geq 2^{-\mathrm{poly}(n)}$, i.e., $B(z, r) \subseteq S(\mathbf{g})$ for some $z \in \mathbb{R}^n$.
\end{enumerate}
Then, for fixed $t \geq \lceil \deg(f)/2\rceil$ and $\varepsilon \geq 2^{-\mathrm{poly}(n)}$, the bound $\mathrm{sos}(f)_t$ can be computed in polynomial time in~$n$ up to an additive error of at most $\varepsilon$. 
\end{theorem}

This suggests that if the feasible set of our polynomial optimization problem \eqref{eq_poly_opt_problem} contains a ball of radius $r \geq 2^{-\poly(n)}$ centered at $z \in \mathbb{R}^n$,
there will be a solution with bit complexity $\poly(n,\log(1/\eps))$, if we add the redundant constraint \eqref{eq:archimed} to make the feasible set explicitly bounded.  
The containment of a ball can be tested on an instant-to-instance basis, cf. \cite{kellner2013containment,kellner2015semidefinite},  
but it seems non-trivial to guarantee in general. 
That is: one could include the corresponding constraint \cite{kellner2013containment,kellner2015semidefinite} 
in the problem formulation, but this could produce an infeasible SDP instance, which is known \cite{Ramana1997} to be challenging. 

Alternatively, one could bound the volume of the semi-algebraic feasible set \cite{henrion2009approximate,korda2018convergence,tacchi2022exploiting},
and then show that semi-algebraic sets with large volume contain a large ball. In the case of convex sets, this can be done utilizing John's theorem, 
cf. Lemma 11 in \cite{gribling2023note}.
Considering that the rates of convergence of volume estimation are known \cite{korda2018convergence},
this may allow for a more detailed analysis. 

\section{Conclusions}

The problem of learning local quantum Hamiltonians is of considerable importance, not least in the implementation
of multi-qubit gates by pulse shaping using quantum optimal control \cite{bondar2023globally}, which requires a quantum Hamiltonians.
One approach to this problem utilizes approximates matrix exponential with polynomials and hierarchies of semidefinite programming relaxations
for polynomial optimization. 
We have developed a novel flat approximation of the matrix exponential utilizing Chebyshev expansion, which strictly improves
upon the use of Taylor expansion in Bakshi et al.\ \cite{bakshi2023learning}.
We have also clarified the issues of bit complexity, which are crucial for proving polynomial run time of algorithms
utilizing moment/SOS relaxations obtained with the flat approximation of the matrix exponential. 
Similar techniques could also be used in the learning of models of open quantum systems \cite{Bondar2020,bondar2022recovering}, and the use of models of quantum systems in
quantum optimal control, following Bondar et al. \cite{bondar2022recovering,bondar2023globally}.
We hope that this will stimulate further research in this direction. 





\paragraph*{Acknowledements}
Ales Wodecki and Jakub Marecek acknowledge the support of the Czech Science Foundation (23-07947S).

\bibliographystyle{ieeetr}
\bibliography{refs}
\end{document}